\documentclass[sn-mathphys,Numbered]{sn-jnl}

\usepackage{graphicx}%
\usepackage{multirow}%
\usepackage{amsmath,amssymb,amsfonts}%
\usepackage{amsthm}%
\usepackage{mathrsfs}%
\usepackage[title]{appendix}%
\usepackage{xcolor}%
\usepackage{textcomp}%
\usepackage{manyfoot}%
\usepackage{booktabs}%
\usepackage{algorithm}%
\usepackage{algorithmicx}%
\usepackage{algpseudocode}%
\usepackage{listings}%

\theoremstyle{thmstyleone}%
\newtheorem{theorem}{Theorem}
\newtheorem{proposition}[theorem]{Proposition}%

\theoremstyle{thmstyletwo}%
\newtheorem{example}{Example}%

\theoremstyle{thmstylethree}%

\newtheorem{lemma}[theorem]{Lemma}
\newtheorem{corollary}[theorem]{Corollary}

\begin{document}

\title[Article Title]{Diffusion Theory of Hyperbolic Groups}


\author*[1]{\fnm{Peter} \sur{Morrison}}\email{peter.morrison@uts.edu.au}

\affil*[1]{\orgdiv{Department of Maths and Physical Sciences}, \orgname{University of Technology, Sydney}, \orgaddress{\street{Broadway}, \city{Sydney}, \postcode{2007}, \state{NSW}, \country{Australia}}}

\abstract{

This paper outlines a method where a brachistochrone is developed
for the hyperbolic plane. This technique is then used to calculate
the Fubini-Study metric and consequent Laplacian operator. We discuss
the various systems of eigenfunctions on the Poincare disk, including
Mehler-Fock, Macdonald and Whittaker functions. The relationship between
these systems of differential equations is exploited using an Laplace
transform method on the hyperbolic plane, which allows us to transform
the solutions to the Helmholz equation from one space to another.
Discussion of further applications of this technique is given with
particular reference to diffusion systems on alternate forms of hyperbolic
and projective geometry.
}




\maketitle

\section{Introduction}

Recent progress in the analysis of diffusion problems for various
systems of special functions has produced a number of interesting
results for transition probability densities, kernels, and Green's
functions. For example, one may find calculations of expressions for
the Laplace transform of the probability kernel in \cite{albanese2007laplace, albanese2007transformations,craddock2014integral}
where two different methods are applied to find the kernel density
related to the CIR process. We seek a more fundamental understanding
of how the various systems given by the integral transform pairs given
by the Mehler-Fock, Kontorovich-Lebedev and Whittaker diffusions come
about, and a deeper theory of how these transform pairs relate to
one another in terms of the solution space of eigenfunctions. We note
that the paper of \cite{sousa2017spectral} covers the Titchmarsh-Kodaira theorem
and related results. This note shall perform calculations that result
in similar formulae, however, we shall deviate significantly from
any of the other methodologies by relying on insights from the theory
of groups and differential geometry. We shall use the intersection
between the metrics and differential geometry which describe continuous
functions on curvilinear spaces, the operators which define translations
on the eigenfunctions themselves, and some simple observations from
the use of hyperbolic group theory to derive a number of interesting
results. We shall be concerned with the continuous states for the
most part in the following computation, but it is useful to remember
that the discrete part of the spectrum is important for certain regimes
of closely related systems to the ones we shall discuss.

The calculation shall proceed by computing the brachistochrone equation
on the hyperbolic surface, the metrics and differential geometry developed
is then used to define a diffusion problem. By examining the various
different transformations between groups of eigenstates, it is then
shown that a number of seemingly disparate systems are effectively
transmutable.

\section{Hyperbolic Brachistochrone}

Let us now talk about how the quantum brachistochrone method may be
extended to hyperbolic space. As is well known, SU(2) is associated
to a spherical geometry; we are interested in how the SU(2)$\leftrightarrow$SU(1,1)
equivalence may be exploited. In general terms, we expect that it
will take the form of a rotation into imaginary time, in the same
way that a Wick rotation \cite{wick1954properties} allows us to take a complex
path integral for a quantum space and use this to generate a real
path integral for a stochastic space, in imaginary time. Writing now
\begin{equation}
    i\left|\dot{\Psi}\right\rangle =i\dfrac{d}{dt}\left|\Psi\right\rangle =\tilde{H}(t)\left|\Psi\right\rangle 
\end{equation}
it is simple to see that we will have the quantum brachistochrone defined through the projective metric:

\begin{theorem}
\begin{equation}
\int1dt=\int\dfrac{\sqrt{\left\langle \Psi\right|\left.\tilde{H}(1-\hat{P})\tilde{H}\right.\left|\Psi\right\rangle }}{\Delta E}dt
\end{equation}
\end{theorem}

or, upon using the standard expression for the projective operator,
i.e. $\hat{P}=\left|\Psi\right\rangle \left\langle \Psi\right|$, the energy dispersion is then
\begin{theorem}
    
\begin{equation}
\Delta E=\sqrt{\left\langle \Psi\right|\left.\tilde{H}^{2}(t)\left|\Psi\right\rangle -\left(\left\langle \Psi\right|\tilde{H}(t)\left|\Psi\right\rangle \right)^{2}\right.}
\end{equation}
\end{theorem}
which we recognise as the standard expression for the variance in
the energy as represented in the quantum system. Let us now examine
the nature of complexification with respect to this system. We have
a quantum state which is dependent on time via $\left|\Psi(t)\right\rangle $,
we have also variously a Hamiltonian matrix $\tilde{H}(t)$ and a
constraint $\tilde{F}(t)$, all of these objects are required in order
to specify the state of the system. The quantum brachistochrone, which
is equivalent to the von Neumann equation for the operator $\tilde{H}(t)+\tilde{F}(t)$,
may be written as \cite{carlini2005quantum, carlini2006time, carlini2007time, morrison2008time, morrison2012time, morrison2019time}
\begin{theorem}

\begin{equation}
i\dfrac{d}{dt}(\tilde{H}(t)+\tilde{F}(t))=\tilde{H}(t)\tilde{F}(t)-\tilde{F}(t)\tilde{H}(t)
\end{equation}
    
\end{theorem}
and the unitary operator of the system is defined by

\begin{lemma}
    
$\hat{U}(t,s)=\exp\left(-i\int_{s}^{t}\tilde{H}(\tau)d\tau\right)$
\end{lemma}.

\begin{proposition}
    It is possible to extend the quantum brachistochrone to hyperbolic systems by using analytic continuation to complex time.
\end{proposition}
\begin{corollary}
    Matrices of group representation theory can be associated with brachistochrones on hyperbolic space.
\end{corollary}

\begin{proof}

The basic hyperbolic system and outline of representation theory may be found in Vilenkin \cite{vilenkin1978special}, where the authors have determined the nature of the unitary matrices used in their analysis
of SU(1,1). Following their analysis, we have e.g.
\begin{example}[Quasiunitary Operator on SU(1,1)]
    
\begin{equation}
\hat{u}_{2}=\left[\begin{array}{cc}
\cosh\dfrac{t}{2} & i\sinh\dfrac{t}{2}\\
\\
-i\sinh\dfrac{t}{2} & \cosh\dfrac{t}{2}
\end{array}\right]
\end{equation}
\end{example}
as a fundamental unitary operator in SU(1,1), whereas previous computation
in \cite{carlini2005quantum, vilenkin1978special} has shown that matrices
of the form

\begin{example}
    
\begin{equation}
\hat{U}_{S}(t,s)=e^{i\phi}\left[\begin{array}{cc}
\cos\phi & \sin\phi\\
-\sin\phi & \cos\phi
\end{array}\right]
\end{equation}
\end{example}
are the relevant concern for SU(2), with $\phi=\omega(t-s)$. Let
us examine the relation between these two groups of matrices. Obviously
we must look at extensions into complex time via $t\rightarrow it$.
    
\end{proof}
If we neglect the global phase in $\hat{U}_{S}(t,s)$, we can see
that 

\begin{lemma}
    
\begin{equation}
\hat{U}_{S}(it,0)=e^{-\omega t}\left[\begin{array}{cc}
\cosh\omega t & i\sinh\omega t\\
-i\sinh\omega t & \cosh\omega t
\end{array}\right]
\end{equation}
\end{lemma}
which, up to scaling by the factor $e^{-\omega t}$, is equivalent
to the matrix $\hat{u}_{2}$ used in Vilenkin \cite{vilenkin1978special}. We
can see the nature of the transformation between spherical and hyperbolic
spaces in this formula. Let us now consider further the implications
for the action principle. Indeed, the quantum brachistochrone, by
its nature, implies optimality in terms of time taken between any
two points on the complex projective space. It seems sensible to then
conjecture that, in order to obtain the hyperbolic counterparts of
the spherical operators, we must look at optimality in terms of imaginary
time. 

\begin{proposition}
    The operators that we obtain from a hyperbolic brachistochrone are not unitary. On the hyperbolic space, one must deal with quasiunitary operators, which have a non-Hermitian symmetry. 
\end{proposition}
\begin{proof} 

If we write $\hat{U}_{S}(it,is)=\hat{\mathcal{U}}_{S}(t,s)$,
we have time translation invariance in the form:

\[
\hat{\mathcal{U}}_{S}(t,s)=\hat{U}_{S}(it,is)=\hat{U}_{S}(i(t-s),0)
\]

\begin{equation}
=e^{-\omega(t-s)}\left[\begin{array}{cc}
\cosh[\omega(t-s)] & i\sinh[\omega(t-s)]\\
-i\sinh[\omega(t-s)] & \cosh[\omega(t-s)]
\end{array}\right]
\end{equation}
Expanding by using the double angle formulae for hyperbolic sines
and cosines, we find:

\begin{lemma}
    
\begin{equation}
\hat{\mathcal{U}}_{S}(t,s)=\hat{\mathcal{U}}_{S}(t,0)\hat{\mathcal{U}}_{S}(0,s)
\end{equation}
with
\begin{equation}
\hat{\mathcal{U}}_{S}(t,0)=e^{-\omega t}\left[\begin{array}{cc}
\cosh\omega t & i\sinh\omega t\\
-i\sinh\omega t & \cosh\omega t
\end{array}\right]
\end{equation}
\end{lemma}

Obviously we have $\hat{\mathcal{U}}_{S}(0,s)=\hat{\mathcal{U}}_{S}(-s,0)$
which means that this space is not a Hermitian space. 

Let us now examine the nature of the time evolution equation in complex time and see what can be learnt. 
as:
    The time evolution operator in the standard Hermitian space may be written in exponential form as
\begin{theorem}

\begin{equation}
\hat{U}_{S}(t,s)=\exp\left(-i\int_{s}^{t}\tilde{H}_{S}(\tau)d\tau\right)=\hat{W}_{S}(t)\hat{W}_{S}^{-1}(s)
\end{equation}
\end{theorem}

with fundamental eigenmatrix:
\begin{equation}
\hat{W}_{S}(t)=\dfrac{1}{\sqrt{2}}\left[\begin{array}{cc}
e^{2i\omega t} & i\\
ie^{2i\omega t} & 1
\end{array}\right]
\end{equation}
and Hamiltonian operator defined by

\begin{equation}
\tilde{H}_{S}(t)=R\left[\begin{array}{cc}
-\cos(2\omega t) & \sin(2\omega t)\\
\sin(2\omega t) & \cos(2\omega t)
\end{array}\right]
\end{equation}

also $\hat{W}_{S}^{-1}(t)=\hat{W}_{S}^{\dagger}(t)$. 

Let us see what happens if we perform the change into complex time. The evolution
operator will be altered in the following way:

\begin{lemma}

\[
\hat{\mathcal{U}}_{S}(t,s)=\hat{U}_{S}(it,is)=\exp\left(-i\int_{is}^{it}\tilde{H}_{S}(\tau)d\tau\right)
\]

\begin{equation}
=\exp\left(-\int_{s}^{t}\tilde{H}_{S}(i\tau)d(i\tau)\right)
\end{equation}
    
\end{lemma}
We would hope that $\hat{U}_{S}(it,is)=\hat{W}_{S}(it)\hat{W}_{S}^{\dagger}(-is)$.
Let us see if this is indeed the case. Calculating the right-hand
side, we have immediately:

\begin{theorem}
    
\begin{equation}
\hat{W}_{S}(it)\hat{W}_{S}^{\dagger}(-is)=e^{-\phi}\left[\begin{array}{cc}
\cosh\phi & i\sinh\phi\\
-i\sinh\phi & \cosh\phi
\end{array}\right]
\end{equation}
\end{theorem}
We notice that this is the same transformation we previously calculated,
with $\phi=\omega(t-s)$. The technique works effectively as can be
seen.    
\end{proof}
\begin{proposition}
    The hyperbolic equivalents of the Hamiltonian operator can also be found by analytic continuation to complex time.
\end{proposition}
\begin{proof}

 In terms of the Hamiltonian operator on the hyperbolic space,
we have

\begin{lemma}
    
\[
\tilde{\mathcal{H}}_{S}(t)=\tilde{H}_{S}(it)=R\left[\begin{array}{cc}
-\cos(2\omega it) & \sin(2\omega it)\\
\sin(2\omega it) & \cos(2\omega it)
\end{array}\right]
\]

\begin{equation}
=R\left[\begin{array}{cc}
-\cosh(2\omega t) & i\sinh(2\omega t)\\
i\sinh(2\omega t) & \cosh(2\omega t)
\end{array}\right]
\end{equation}

\end{lemma}

\end{proof}

Let us now consider the effect of this type of transformation on the
quantum brachistochrone equations. This set of coupled DEs is described
by the system
\begin{equation}
i\dfrac{d}{dt}(\tilde{H}(t)+\tilde{F}(t))=[\tilde{H}(t),\tilde{F}(t)]
\end{equation}
Let us examine the nature of differentiating e.g. $\dfrac{d}{dt}(\tilde{H}(it)+\tilde{F}(it))$,
the left-hand side of the above equation will then be changed to:
\begin{equation}
-\dfrac{d}{dt}(\tilde{H}(it)+\tilde{F}(it))=[\tilde{H}(it),\tilde{F}(it)]
\end{equation}
We may further transform this equation by a constant matrix $\hat{S}$,
and we would hope to find the set of coupled differential equations
defined by:
\begin{proposition}
\begin{equation}
-\dfrac{d}{dt}(\tilde{\mathcal{H}}_{S}(t)+\tilde{\mathcal{F}}_{S}(t))=[\tilde{\mathcal{H}}_{S}(t),\tilde{\mathcal{F}}_{S}(t)]
\end{equation}
\end{proposition}
\begin{proof}
For our matrix
$\tilde{\mathcal{H}}_{S}(t)$, we have associated constraint law $\tilde{\mathcal{F}}_{S}(t)=\Omega\hat{\sigma}_{y}$,
obviously these two operators satisfy $\mathrm{Tr}(\tilde{\mathcal{H}}_{S}(t)\tilde{\mathcal{F}}_{S}(t))=0$,
. We have $\tilde{\mathcal{F}}_{S}(t)=\tilde{\mathcal{F}}_{0}$,
so we would hope that we are able to satisfy
\begin{equation}
-\dfrac{d\tilde{\mathcal{H}}_{S}(t)}{dt}=[\tilde{\mathcal{H}}_{S}(t),\tilde{\mathcal{F}}_{0}]
\end{equation}
Calculating the left-hand side, we have 
\begin{equation}
-\dfrac{d\tilde{\mathcal{H}}_{S}(t)}{dt}=-2\omega R\left[\begin{array}{cc}
-\sinh(2\omega t) & i\cosh(2\omega t)\\
i\cosh(2\omega t) & \sinh(2\omega t)
\end{array}\right]
\end{equation}
whereas the right-hand side of this expression reads as:
\begin{equation}
[\tilde{\mathcal{H}}_{S}(t),\tilde{\mathcal{F}}_{0}]=2\Omega R\left[\begin{array}{cc}
-\sinh(2\omega t) & i\cosh(2\omega t)\\
i\cosh(2\omega t) & \sinh(2\omega t)
\end{array}\right]
\end{equation}
hence we have a solution for $\Omega=-\omega$. Q.E.D.
    
\end{proof}
 This is the appropriate
method to use to calculate a brachistochrone in a hyperbolic space.
\begin{proposition}
    The isotropic constraint is modified in a hyperbolic system from:
    
\begin{equation}
\mathrm{Tr}\left(\dfrac{\tilde{H}^{2}}{2}\right)=k<\infty
\end{equation}
to
\begin{equation}
\mathrm{Tr}\left(\dfrac{\tilde{H}^{2}}{2}\right)=-R^{2}>-\infty
\end{equation}  

\end{proposition}
\begin{proof}

If we take the formula we derived above, we have 
\begin{equation}
\tilde{\mathcal{H}}_{S}(t)=\tilde{H}_{S}(it)=R\left[\begin{array}{cc}
-\cosh(2\omega t) & i\sinh(2\omega t)\\
i\sinh(2\omega t) & \cosh(2\omega t)
\end{array}\right]
\end{equation}
Naiively computing the isotropic constraint, we obtain:
\begin{equation}
\mathrm{Tr}\left(\dfrac{\tilde{H}^{2}}{2}\right)=-R^{2}[\cosh^{2}(2\omega t)-\sinh^{2}(2\omega t)]=-R^{2}>-\infty
\end{equation}      
\end{proof}
We can see that this method will be effective for finding brachistochrones
on these types of quasi-unitary spaces. Finally, we shall comment
on the nature of the action principle in a hyperbolic space. Following
Carlini et. al \cite{carlini2005quantum, carlini2006time, carlini2007time}, for a standard
time optimal quantum system, the energy dispersion is written as:
\begin{equation}
\Delta E=\sqrt{\left\langle \Psi\right|\left.\tilde{H}^{2}(t)\left|\Psi\right\rangle -\left(\left\langle \Psi\right|\tilde{H}(t)\left|\Psi\right\rangle \right)^{2}\right.}
\end{equation}
We have already shown that the trace conditions and isotropic relation
are essentially unchanged by moving to complex time. We are posed,
therefore, with the question of what exactly constitutes the state
vector in this type of hyperbolic space. One way to address this question is through the use of the ideas of projective geometry.

\section{Evaluation of Fubini-Study Metric for SU(1,1)}

Let us consider one of the unitary operators derived in the SU(1,1)
calculation. We would like to show a quick and simple method for showing
the link between the unitary operator which defines the evolution,
and the metric which defines the Laplace operator. If we can do so,
it will make many of the intermediate steps in the calculations in
\cite{vilenkin1978special} fall out as a natural consequence. We have unitary
operator given by:

\begin{lemma}
\begin{equation}
\hat{W}_{S}(it)\hat{W}_{S}^{\dagger}(-is)=e^{-\phi}\left[\begin{array}{cc}
\cosh\phi & i\sinh\phi\\
-i\sinh\phi & \cosh\phi
\end{array}\right]=e^{-\phi}\hat{\omega}_{2}(\phi)
\end{equation}
\end{lemma}
For now, removing the scaling as superfluous to the dynamics on the
hyperbolic plane, we can examine purely the matrix operator, which
may be expanded thus:

\begin{equation}
\hat{\omega}_{2}(t+\tau)=\left[\begin{array}{cc}
\cosh(t+\tau) & i\sinh(t+\tau)\\
-i\sinh(t+\tau) & \cosh(t+\tau)
\end{array}\right]
\end{equation}
We know that it is possible to use a double angle formula to write:
\begin{equation}
\hat{\omega}_{2}(t+\tau)=\left[\begin{array}{cc}
\cosh(t)\cosh(\tau)+\sinh(t)\sinh(\tau) & i\left(\sinh(t)\cosh(\tau)+\cosh(t)\sinh(\tau)\right)\\
-i\left(\sinh(t)\cosh(\tau)+\cosh(t)\sinh(\tau)\right) & \cosh(t)\cosh(\tau)+\sinh(t)\sinh(\tau)
\end{array}\right]
\end{equation}
and therefore the unitary operator composes as 
\begin{lemma}
    
$\hat{\omega}_{2}(t+\tau)=\hat{\omega}_{2}(\tau)\hat{\omega}_{2}(t)$
\end{lemma}
as we would expect in terms of the group representation and composition identities. In terms of the total rotation, we therefore have
generators for SU(1,1):
\begin{lemma}
    
\begin{equation}
\hat{a}_{+}=\dfrac{1}{2}\left[\begin{array}{cc}
0 & 1\\
1 & 0
\end{array}\right]
\end{equation}
\begin{equation}
\hat{a}_{-}=\dfrac{1}{2}\left[\begin{array}{cc}
0 & i\\
-i & 0
\end{array}\right]
\end{equation}
\begin{equation}
\hat{a}_{3}=\dfrac{i}{2}\left[\begin{array}{cc}
1 & 0\\
0 & -1
\end{array}\right]
\end{equation}
\end{lemma}
and we have the SU(1,1) algebra
\begin{corollary}
    
$[\hat{a}_{\pm},\hat{a}_{3}]=\pm\hat{a}_{\mp}$,
$[\hat{a}_{+},\hat{a}_{-}]=\hat{a}_{3}$
\end{corollary}
\begin{proposition}
The hyperbolic state may be evaluated by using the Euler decomposition.
\end{proposition}
\begin{proof}
We have access to the eigenstate
of the entire system, evaluation of the time evolution operator using
the Euler decomposition gives:
\begin{equation}
\hat{U}(\tau,\varphi,\psi)=\hat{\omega}_{3}(\varphi)\hat{\omega}_{2}(\tau)\hat{\omega}_{3}(-\psi)
\end{equation}
\begin{equation}
=\left[\begin{array}{cc}
e^{i\varphi/2} & 0\\
0 & e^{-i\varphi/2}
\end{array}\right]\left[\begin{array}{cc}
\cosh\dfrac{\tau}{2} & i\sinh\dfrac{\tau}{2}\\
-i\sinh\dfrac{\tau}{2} & \cosh\dfrac{\tau}{2}
\end{array}\right]\left[\begin{array}{cc}
e^{-i\psi/2} & 0\\
0 & e^{i\psi/2}
\end{array}\right]
\end{equation}
or
\begin{equation}
\hat{U}(\tau,\varphi,\psi)=\left[\begin{array}{cc}
\cosh\dfrac{\tau}{2}e^{i(\varphi-\psi)/2} & i\sinh\dfrac{\tau}{2}e^{i(\varphi+\psi)/2}\\
-i\sinh\dfrac{\tau}{2}e^{-i(\varphi+\psi)/2} & \cosh\dfrac{\tau}{2}e^{-i(\varphi-\psi)/2}
\end{array}\right]
\end{equation}
Applying this to an initial state, we have:
\begin{equation}
\left|\Psi(\tau,\varphi,\psi)\right\rangle =\hat{U}(\tau,\varphi,\psi)\left|\Psi(0)\right\rangle =\left[\begin{array}{cc}
\cosh\dfrac{\tau}{2}e^{i(\varphi+\psi)/2} & i\sinh\dfrac{\tau}{2}e^{i(\varphi-\psi)/2}\\
\\
-i\sinh\dfrac{\tau}{2}e^{-i(\varphi-\psi)/2} & \cosh\dfrac{\tau}{2}e^{-i(\varphi+\psi)/2}
\end{array}\right]\left[\begin{array}{c}
a\\
\\
b
\end{array}\right]
\end{equation}
\begin{equation}
=\left[\begin{array}{c}
a\cosh\dfrac{\tau}{2}e^{i(\varphi-\psi)/2}+ib\sinh\dfrac{\tau}{2}e^{i(\varphi+\psi)/2}\\
\\
b\cosh\dfrac{\tau}{2}e^{-i(\varphi-\psi)/2}-ia\sinh\dfrac{\tau}{2}e^{-i(\varphi+\psi)/2}
\end{array}\right]
\end{equation}    
\end{proof}

\begin{proposition}
    In the hyperbolic state space, the Fubini-Study metric is given by the Poincare metric.
\end{proposition}
\begin{proof}
    
Let us consider the Fubini-Study metric. We know that there will be
a relationship with the metric familiar from differential geometry.
In the space of states, we have expectation values:
\begin{theorem}
\begin{equation}
\left\langle \Psi(\tau,\varphi,\psi)\right|\left.\hat{A}_{m}\hat{A}_{n}\right.\left|\Psi_{i}(0)\right\rangle =\left\langle \hat{A}_{m}\hat{A}_{n}\right\rangle 
\end{equation}
\end{theorem}
We may use this relationship to project out each component in turn.
Let us do so for each of the operators that define the algebra. The
Fubini-Study metric tensor is given by:
\begin{theorem}
    \begin{equation}
g_{\alpha\beta}=\mathfrak{Re}\left[\left\langle \bar{\psi}_{\alpha}\right|\left.\psi_{\beta}\right\rangle -\left\langle \bar{\psi}_{\alpha}\right|\left.\Psi\right\rangle \left\langle \Psi\right|\left.\psi_{\beta}\right\rangle \right]
\end{equation}
\begin{equation}
\left|\psi_{\beta}\right\rangle =\dfrac{\partial}{\partial x_{\beta}}\left|\Psi\right\rangle 
\end{equation}
\end{theorem}

\begin{theorem}
    
For our particular choice of quasiunitary operator, we have
the adjoint equation given by:
\begin{equation}
\hat{U}^{-1}(\tau,\varphi,\psi)=\hat{U}^{T}(\tau,-\varphi,-\psi)
\end{equation}
such that $\left\langle \bar{\Psi}(\tau,\varphi,\psi)\right|=\left\langle \bar{\Psi}(0)\right|\hat{U}^{T}(\tau,-\varphi,-\psi)=\left\langle \Psi(\tau,-\varphi,-\psi)\right|$.
\end{theorem}

We make special note that in this space the bra-ket notation signifies
only transposition and not the composition of conjugation, as this
is a hyperbolic space. Evaluating, we then have differential state
components:
\begin{equation}
\left|\Psi\right\rangle =\left|\Psi(\tau,\varphi,\psi)\right\rangle =\hat{U}(\tau,\varphi,\psi)\left|\Psi(0)\right\rangle 
\end{equation}
If we take the initial state to be given by $\left[\begin{array}{c}
1\\
0
\end{array}\right],$we may read off the necessary variables as:
\begin{lemma}
    
\begin{equation}
\left|\Psi\right\rangle =\left[\begin{array}{c}
\cosh\dfrac{\tau}{2}e^{i(\varphi+\psi)/2}\\
\\
-i\sinh\dfrac{\tau}{2}e^{-i(\varphi-\psi)/2}
\end{array}\right]
\end{equation}
\end{lemma}
we then have:
\begin{lemma}
    
\begin{equation}
\left|\Psi_{\tau}\right\rangle =\dfrac{\partial}{\partial\tau}\left|\Psi\right\rangle =\dfrac{1}{2}\left[\begin{array}{c}
\sinh\dfrac{\tau}{2}e^{i(\varphi+\psi)/2}\\
\\
-i\cosh\dfrac{\tau}{2}e^{-i(\varphi-\psi)/2}
\end{array}\right]
\end{equation}
\begin{equation}
\left|\Psi_{\varphi}\right\rangle =\dfrac{\partial}{\partial\varphi}\left|\Psi\right\rangle =\dfrac{1}{2}\left[\begin{array}{c}
i\cosh\dfrac{\tau}{2}e^{i(\varphi+\psi)/2}\\
\\
-\sinh\dfrac{\tau}{2}e^{-i(\varphi-\psi)/2}
\end{array}\right]
\end{equation}
\begin{equation}
\left|\Psi_{\psi}\right\rangle =\dfrac{\partial}{\partial\psi}\left|\Psi\right\rangle =\dfrac{1}{2}\left[\begin{array}{c}
i\cosh\dfrac{\tau}{2}e^{i(\varphi+\psi)/2}\\
\\
\sinh\dfrac{\tau}{2}e^{-i(\varphi-\psi)/2}
\end{array}\right]
\end{equation}
\end{lemma}
Using the formula for the adjoint $\left\langle \bar{\Psi}(\tau,\varphi,\psi)\right|=\left|\Psi(\tau,-\varphi,-\psi)\right\rangle ^{T}$,
we may now assemble all the necessary parts of the Fubini-Study metric.
The projection operator may be written:
\begin{theorem}
    
\begin{equation}
\hat{P}(\tau,\varphi,\psi)=\left|\Psi(\tau,\varphi,\psi)\right.\left\rangle \right\langle \left.\bar{\Psi}(\tau,\varphi,\psi)\right|=\left[\begin{array}{cc}
\cosh^{2}\dfrac{\tau}{2} & -ie^{i\varphi}\cosh\dfrac{\tau}{2}\sinh\dfrac{\tau}{2}\\
\\
-ie^{-i\varphi}\cosh\dfrac{\tau}{2}\sinh\dfrac{\tau}{2} & -\sinh^{2}\dfrac{\tau}{2}
\end{array}\right]
\end{equation}
\end{theorem}

with $\mathrm{Tr}\hat{P}=1$.
 The formula for the Fubini-Study metric
reads as:
\begin{equation}
F_{\alpha\beta}=\left\langle \bar{\Psi}_{\alpha}\right|\left.\Psi_{\beta}\right\rangle -\left\langle \bar{\Psi}_{\alpha}\right|\left.\Psi\right\rangle \left\langle \Psi\right|\left.\Psi_{\beta}\right\rangle 
\end{equation}
\begin{equation}
g_{\alpha\beta}=\mathfrak{Re}\left[F_{\alpha\beta}\right]
\end{equation}
Computing, we find the tensor:
\begin{theorem}

\begin{equation}
F_{\alpha\beta}=\dfrac{1}{4}\left[\begin{array}{ccc}
-1 & i\sinh\tau & 0\\
i\sinh\tau & \dfrac{1}{2}(\cosh2\tau-1) & 0\\
0 & 0 & 0
\end{array}\right]
\end{equation}
or
\begin{equation}
g_{\alpha\beta}=\dfrac{1}{4}\left[\begin{array}{ccc}
-1 & 0 & 0\\
0 & \sinh^{2}\tau & 0\\
0 & 0 & 0
\end{array}\right]
\end{equation}
    
\end{theorem}
We therefore derive the second fundamental form to be equal to:
\begin{theorem}
    
\begin{equation}
ds^{2}=\dfrac{1}{4}(-d\tau^{2}+\sinh^{2}\tau d\varphi^{2})
\end{equation}
\end{theorem}
which is equal to the metric for the hyperbolic plane model of Poincare.
\end{proof}
We can see how the theorems from the previous section that derived
the metric in a curvilinear space may be immediately applied, and
how a change to the dimension of the sphere will add an extra radial
component into the equations as stated. We can see that this method
of analysis is much simpler in nature than the method of group representations.
By moving directly from the unitary transformation to a Laplacian
operator with a symmetry group of the same type, we remove all the
difficulties of solving coupled equations. We shall now discuss the
results and unsolved problems which are related to this analysis.

\section{Calculation of Hyperbolic Laplacian}
\begin{proposition}
    The hyperbolic Laplacian associated to the Fubini-Study metric describes a set of eigenfunctions that are given by the Mehler-Fock functions of the hyperbolic distance.
\end{proposition}
\begin{proof}
    
Let us consider first the parametrisation of the coordinates on the
hyperbolic surface. The metric coming from the previous calculation
may be written as:
\begin{equation}
4ds^{2}=g_{\alpha\beta}dx^{\alpha}dx^{\beta}
\end{equation}
\begin{equation}
g_{\alpha\beta}=\left[\begin{array}{cc}
-1 & 0\\
0 & \sinh^{2}\tau
\end{array}\right]
\end{equation}
The Laplace-Beltrami operator in a curved space is given by:
\begin{equation}
\nabla^{2}f=\dfrac{1}{\sqrt{g}}\dfrac{\partial}{\partial x_{\alpha}}\left(\sqrt{g}g^{\alpha\beta}\dfrac{\partial f}{\partial x_{\beta}}\right)
\end{equation}
We have e.g. $g=\det g_{\alpha\beta}=-\sinh^{2}\tau$, and contravariant
metric given by:
\begin{equation}
g^{\alpha\beta}=\left[\begin{array}{cc}
-1 & 0\\
0 & \dfrac{1}{\sinh^{2}\tau}
\end{array}\right]
\end{equation}
so we conclude that the Laplace operator is given the formula:
\begin{equation}
\nabla^{2}f=\dfrac{1}{i\sinh\tau}\sum_{\alpha,\beta}\dfrac{\partial}{\partial x_{\alpha}}\left(i\sinh\tau g^{\alpha\beta}\dfrac{\partial f}{\partial x_{\beta}}\right)=\dfrac{1}{\sinh\tau}\sum_{\alpha,\beta}\dfrac{\partial}{\partial x_{\alpha}}\left(\sinh\tau g^{\alpha\beta}\dfrac{\partial f}{\partial x_{\beta}}\right)
\end{equation}
The only non-zero components are $g^{\tau\tau}$ and $g^{\varphi\varphi}$,
hence we find:
\begin{equation}
\nabla^{2}f=\dfrac{1}{\sinh\tau}\dfrac{\partial}{\partial\tau}\left(\sinh\tau.g^{\tau\tau}\dfrac{\partial f}{\partial\tau}\right)+\dfrac{1}{\sinh\tau}\dfrac{\partial}{\partial\varphi}\left(\sinh\tau.g^{\varphi\varphi}\dfrac{\partial f}{\partial\varphi}\right)
\end{equation}
\begin{equation}
=-\dfrac{1}{\sinh\tau}\dfrac{\partial}{\partial\tau}\left(\sinh\tau\dfrac{\partial f}{\partial\tau}\right)+\dfrac{1}{\sinh^{2}\tau}\dfrac{\partial^{2}f}{\partial\varphi^{2}}
\end{equation}
\begin{equation}
=-\dfrac{\partial^{2}f}{\partial\tau^{2}}-\coth\tau\dfrac{\partial f}{\partial\tau}+\dfrac{1}{\sinh^{2}\tau}\dfrac{\partial^{2}f}{\partial\varphi^{2}}
\end{equation}

We may associate the following standard differential systems with the curvilinear Laplace operator, firstly the solution of the Laplace equation:
\begin{equation}
\nabla^{2}f=0
\end{equation}The Helmholtz equation:
\begin{equation}
\nabla^{2}f=\lambda f
\end{equation}
The quantum equation:
\begin{equation}
E\Psi=-\nabla^{2}\Psi+V(\tau,\varphi)\Psi
\end{equation}
Writing out the last of these explicitly for the Helmholtz case i.e.
zero potential, we have:

\begin{corollary}
    
\begin{equation}
E\Psi=-\dfrac{\partial^{2}\Psi}{\partial\tau^{2}}-\coth\tau\dfrac{\partial\Psi}{\partial\tau}+\dfrac{1}{\sinh^{2}\tau}\dfrac{\partial^{2}\Psi}{\partial\varphi^{2}}
\end{equation}
\end{corollary}
Now, consider the solution $\Psi(\tau,\varphi)$. We may Fourier transform
in the second variable, writing $\Upsilon(\tau,k)=\int_{-\infty}^{+\infty}e^{-ik\varphi}\Psi(\tau,\varphi)d\varphi$.
Then the differential equation becomes:
\begin{equation}
E\Upsilon=-\dfrac{\partial^{2}\Upsilon}{\partial\tau^{2}}-\coth\tau\dfrac{\partial\Upsilon}{\partial\tau}-\dfrac{k^{2}}{\sinh^{2}\tau}\Upsilon
\end{equation}
\begin{equation}
\dfrac{\partial^{2}\Upsilon}{\partial\tau^{2}}+\coth\tau\dfrac{\partial\Upsilon}{\partial\tau}+\left(\dfrac{k^{2}}{\sinh^{2}\tau}-E\right)\Upsilon=0
\end{equation}
The solution to this differential equation at fixed $k$ is given
by:
\begin{equation}
\Upsilon(\tau,k)=\dfrac{1}{\sqrt{\sinh\tau}}\left(C_{1}\mathcal{P}_{ik-1/2}^{\sqrt{4E+1}/2}\left(\coth\tau\right)+C_{2}\mathcal{Q}_{ik-1/2}^{\sqrt{4E+1}/2}\left(\coth\tau\right)\right)
\end{equation}
Therefore, by choosing $E=-\dfrac{1}{4}+\xi^{2}$, we obtain:
\begin{theorem}[Hyperbolic Eigenfunctions]
\begin{equation}
\Upsilon(\tau,k)=\dfrac{1}{\sqrt{\sinh\tau}}\left(C_{1}\mathcal{P}_{ik-1/2}^{\xi}\left(\coth\tau\right)+C_{2}\mathcal{Q}_{ik-1/2}^{\xi}\left(\coth\tau\right)\right)
\end{equation}
    
\end{theorem}

and further using the Whipple formulae. 
\begin{equation}
\dfrac{A_{k\rho}}{\sqrt{\sinh\tau}}\mathcal{P}_{k-1/2}^{i\rho}(\coth\tau)=\mathcal{Q}_{i\rho-1/2}^{k}(\cosh\tau)
\end{equation}
\begin{equation}
\dfrac{B_{k\rho}}{\sqrt{\sinh\tau}}\mathcal{Q}_{k-1/2}^{i\rho}(\coth\tau)=\mathcal{P}_{i\rho-1/2}^{k}(\cosh\tau)
\end{equation}
\begin{equation}
\Upsilon(\tau,k)=\mathcal{C}_{1}\mathcal{P}_{\xi-1/2}^{ik}\left(\cosh\tau\right)+\mathcal{C}_{2}\mathcal{Q}_{\xi-1/2}^{ik}\left(\cosh\tau\right)
\end{equation}
This satisfies the differential equation by inspection. It is obvious
to see how the various different representations may be linked up.
Let us calculate the form of the completeness relationship, and then
give the solution to the original differential equation. 

If we take
$k\rightarrow-ik$, $\xi\rightarrow i\xi$, we get differential equation:
\begin{theorem}
    \begin{equation}
\dfrac{\partial^{2}\Upsilon}{\partial\tau^{2}}+\coth\tau\dfrac{\partial\Upsilon}{\partial\tau}+\left(-\dfrac{k^{2}}{\sinh^{2}\tau}+\left(\dfrac{1}{4}+\xi^{2}\right)\right)\Upsilon=0
\end{equation}
\end{theorem} 

with solutions:
\begin{equation}
\Upsilon(\tau,k)=\mathcal{C}_{1}\mathcal{P}_{i\xi-1/2}^{k}\left(\cosh\tau\right)+\mathcal{C}_{2}\mathcal{Q}_{i\xi-1/2}^{k}\left(\cosh\tau\right)
\end{equation} Q.E.D.

\end{proof}

We can see that there will be various forms of solution to the original
differential equation, which will be given by Fourier inversion:
\begin{equation}
\Psi(\tau,\varphi)=\int_{-\infty}^{+\infty}dk.e^{-ik\varphi}\int\Upsilon_{\xi}(\tau,k)d\mathfrak{m}(\xi)
\end{equation}
and the kernels and Green's functions will be of the form:
\begin{equation}
K(x,y;t)=\int_{-\infty}^{+\infty}dk.e^{-ik\varphi}\int e^{-(1/4+\xi^{2})t}\Upsilon_{\xi}(x,k)\Upsilon_{\xi}(y,k)d\mathfrak{m}(\xi)
\end{equation}
\begin{equation}
G(x,y;E)=\int e^{-Et}K(x,y;t)dt
\end{equation}
For the last system, we have a completeness relationship from Grosche
\& Steiner \cite{grosche1987path, grosche1990path, grosche1998handbook}:
\begin{equation}
\left|\dfrac{\Gamma\left(1/2+ip-\mu\right)}{\Gamma(ip)}\right|^{2}\int_{1}^{\infty}\mathcal{P}_{ip-1/2}^{\mu}\left(y\right)\mathcal{P}_{-ip'-1/2}^{\mu}\left(y\right)dy=\delta(p-p')
\end{equation}
\begin{equation}
\int_{0}^{\infty}\left|\dfrac{\Gamma\left(1/2+ip-\mu\right)}{\Gamma(ip)}\right|^{2}\mathcal{P}_{ip-1/2}^{\mu}\left(x\right)\mathcal{P}_{ip-1/2}^{\mu}\left(y\right)dp=\delta(x-y)
\end{equation}

We shall now discuss some related systems of differential equations.
The hyperbolic plane or disk can be seen as a projection from a higher
dimensional object into the plane via a homomorphism. A simple way
of understanding this is in terms of projection from a sphere into
the plane. Each line of latitude gives a circle, which is unique for
the lower half-sphere. In this way, one may identify a point on the
lower half-sphere with a circle and a rotational angle in the plane.
We shall be using this sort of analogy in order to show some further
links between special functions and these types of transformation
laws. By using the completeness relationship between the various sets
of eigenfunctions, we will be able to develop a detailed understanding
of this system. We have shown the underlying link between the geometry
of SU(1,1) and the sets of special functions which are eigenfunctions
of the Laplacian operator. By using the Fubini-Study metric, and observing
that the brachistochrone is geodesic, one can write down the formulae
for the fundamental form, and hence the Laplacian. This is a powerful
technique which bypasses many of the technical difficulties associated
with applying representation theory directly. Indeed, the representation
theory can be seen as derivate, albeit in a rotated frame, to the
existence of a well-ordered geodesic system as described on the manifold.

\section{Related Groups}

\begin{proposition}
    The eigenfunctions on the pseudosphere can be described through an identical process using a hyperbolic metric.
\end{proposition}
\begin{proof}

Let us consider first the parametrisation of the coordinates on the
hyperbolic surface. If we think about the basic equation of the pseudosphere,
we will have an expression such as $z^{2}-x^{2}-y^{2}=r^{2}$, in
which case we may take the expression of a point on the surface as:
\begin{equation}
\begin{array}{c}
x=r\sinh\tau\cos\varphi\\
y=r\sinh\tau\sin\varphi\\
z=r\cosh\tau
\end{array}
\end{equation}
in which case it is a simple exercise in differentiation to show that
the correct form of the metric may be written:
\begin{lemma}
    
\begin{equation}
ds^{2}=g_{\alpha\beta}dx^{\alpha}dx^{\beta}=dz^{2}-dx^{2}-dy^{2}=dr^{2}-r^{2}d\tau^{2}-r^{2}\sinh^{2}\tau d\varphi^{2}
\end{equation}
\end{lemma}
and we conclude that the metric is given by:
\begin{lemma}
    
\begin{equation}
g_{\alpha\beta}=\left[\begin{array}{ccc}
1 & 0 & 0\\
0 & -r^{2} & 0\\
0 & 0 & -r^{2}\sinh^{2}\tau
\end{array}\right]
\end{equation}
\end{lemma}
The Laplace-Beltrami operator in a curved space is given by:
\begin{equation}
\nabla^{2}f=\dfrac{1}{\sqrt{g}}\dfrac{\partial}{\partial x_{\alpha}}\left(\sqrt{g}g^{\alpha\beta}\dfrac{\partial f}{\partial x_{\beta}}\right)
\end{equation}
Now, it is simple to show that we have e.g. $g=\det g_{\alpha\beta}=r^{4}\sinh^{2}\tau$,
and contravariant metric given by:
\begin{equation}
g^{\alpha\beta}=\left[\begin{array}{ccc}
1 & 0 & 0\\
0 & -\dfrac{1}{r^{2}} & 0\\
0 & 0 & -\dfrac{1}{r^{2}\sinh^{2}\tau}
\end{array}\right]
\end{equation}
so we conclude that the Laplace operator is given the formula:
\begin{equation}
\nabla^{2}f=\dfrac{1}{r^{2}\sinh\tau}\sum_{\alpha,\beta}\dfrac{\partial}{\partial x_{\alpha}}\left(r^{2}\sinh\tau g^{\alpha\beta}\dfrac{\partial f}{\partial x_{\beta}}\right)
\end{equation}
The only non-zero components are $g^{rr}$, $g^{\tau\tau}$ and $g^{\varphi\varphi}$,
hence we find:
\begin{equation}
\nabla^{2}f=\dfrac{1}{r^{2}\sinh\tau}\left(\dfrac{\partial}{\partial r}\left(r^{2}\sinh\tau\dfrac{\partial f}{\partial r}\right)-\dfrac{\partial}{\partial\tau}\left(\dfrac{r^{2}\sinh\tau}{r^{2}}\dfrac{\partial f}{\partial\tau}\right)-\dfrac{\partial}{\partial\varphi}\left(\dfrac{r^{2}\sinh\tau}{r^{2}\sinh^{2}\tau}\dfrac{\partial f}{\partial\varphi}\right)\right)
\end{equation}
\begin{equation}
=\dfrac{1}{r^{2}}\dfrac{\partial}{\partial r}\left(r^{2}\dfrac{\partial f}{\partial r}\right)-\dfrac{1}{r^{2}\sinh\tau}\dfrac{\partial}{\partial\tau}\left(\sinh\tau\dfrac{\partial f}{\partial\tau}\right)-\dfrac{1}{r^{2}\sinh^{2}\tau}\dfrac{\partial}{\partial\varphi}\left(\dfrac{\partial f}{\partial\varphi}\right)
\end{equation}
Q.E.D.
\end{proof}

One can see how to develop out of various different potentials the
relevant groups of Bessel and Mehler-Fock functions. This is the basis
for the SU(1,1)\textasciitilde SO(2,1) isomorphism. This paper is
not concerned with this primarily and we shall leave it for a future
discussion. We note that the existence of such Laplacian operators
naturally implies the question of the existence of a state vector,
which in this case would generate a (2+1) dimensional vector representation
of the hyperbolic Hamiltonian, in the same way that SU(3) contains
an SO(3) subgroup, one would hope this would be of the type SU(1,1)\textasciitilde SO(2,1)$\subseteq$SU(2,1)
vs SU(2)\textasciitilde SO(3)$\subseteq$SU(3).

\section{Alternate Representations of the Hyperbolic Plane}
\begin{proposition}
    The Liouville form of the hyperbolic equation may be found using transform theory. This may be further transformed to yield the confluent hypergeometric functions of Whittaker type.
\end{proposition}
\begin{proof}

We shall now show some simple ways in which we can develop further
representations of the hyperbolic plane, using the differential equation
of Legendre type. If we take the fundamental equation, we have:
\begin{lemma}
   \begin{equation}
\dfrac{\partial^{2}f}{\partial\tau^{2}}+\coth\tau\dfrac{\partial f}{\partial\tau}+\left(-\dfrac{k^{2}}{\sinh^{2}\tau}+\left(\dfrac{1}{4}+\xi^{2}\right)\right)f=0
\end{equation} 
\end{lemma}

Now, a simple method to remove the middle term is to use the following
transformation.
\begin{equation}
u''+a_{1}(x)u'+a_{2}(x)u=0
\end{equation}
\begin{equation}
U''+U\left[\phi'+\phi^{2}+\phi a_{1}+a_{2}\right]=0
\end{equation}
\begin{equation}
\phi(x)=-\dfrac{a_{1}(x)}{2}
\end{equation}
\begin{equation}
u(x)=e^{\int\phi dx}U(x)
\end{equation}  
We obtain the heat type equation:
\begin{equation}
\dfrac{\partial^{2}F}{\partial\tau^{2}}+\left(\dfrac{\partial}{\partial\tau}\left[-\dfrac{\coth\tau}{2}\right]+\left[-\dfrac{\coth\tau}{2}\right]^{2}-\dfrac{1}{2}(\coth\tau)^{2}+\left(-\dfrac{k^{2}}{\sinh^{2}\tau}+\left(\dfrac{1}{4}+\xi^{2}\right)\right)\right)F=0
\end{equation}
or
\begin{equation}
\dfrac{\partial^{2}F}{\partial\tau^{2}}+\left(-\dfrac{k^{2}}{\sinh^{2}\tau}+\dfrac{1}{4}\left(\coth^{2}\tau-1\right)+\xi^{2}\right)F
\end{equation}
\begin{equation}
    =\dfrac{\partial^{2}F}{\partial\tau^{2}}+\left(\dfrac{(1/4-k^{2})}{\sinh^{2}\tau}+\xi^{2}\right)F=0
\end{equation}

Now we can also write this equation in a complementary form using
the hyperbolic substitutions. For the original equation, we have:

\begin{equation}
\dfrac{\partial^{2}f}{\partial\tau^{2}}+\coth\tau\dfrac{\partial f}{\partial\tau}+\left(-\dfrac{k^{2}}{\sinh^{2}\tau}+\left(\dfrac{1}{4}+\xi^{2}\right)\right)f=0
\end{equation}
Multiplying out, we have:
\begin{equation}
\sinh^{2}\tau\dfrac{\partial^{2}f}{\partial\tau^{2}}+\sinh\tau\cosh\tau\dfrac{\partial f}{\partial\tau}+\left(-k^{2}+\left(\dfrac{1}{4}+\xi^{2}\right)\sinh^{2}\tau\right)f=0
\end{equation}
Using $z=\cosh\tau$, $z^{2}-\sinh^{2}\tau=1$, we can rewrite:
\begin{equation}
\dfrac{\partial f}{\partial\tau}=\dfrac{\partial f}{\partial z}\dfrac{\partial z}{\partial\tau}=\sinh\tau\dfrac{\partial f}{\partial z}=\sqrt{z^{2}-1}\dfrac{\partial f}{\partial z}
\end{equation}
\begin{equation}
\dfrac{\partial f}{\partial z}=\dfrac{1}{\sinh\tau}\dfrac{\partial f}{\partial\tau}
\end{equation}
\begin{equation}
\dfrac{\partial^{2}f}{\partial z^{2}}=\dfrac{1}{\sinh\tau}\dfrac{\partial}{\partial\tau}\dfrac{1}{\sinh\tau}\dfrac{\partial f}{\partial\tau}=\dfrac{1}{\sinh^{2}\tau}\dfrac{\partial^{2}f}{\partial\tau^{2}}-\left(\dfrac{\cosh\tau}{\sinh^{3}\tau}\right)\dfrac{\partial f}{\partial\tau}
\end{equation}
\begin{equation}
=\dfrac{1}{\sinh^{2}\tau}\left(\dfrac{\partial^{2}f}{\partial\tau^{2}}-\coth\tau\dfrac{\partial f}{\partial\tau}\right)
\end{equation}
\begin{equation}
(z^{2}-1)\dfrac{\partial^{2}f}{\partial z^{2}}=\dfrac{\partial^{2}f}{\partial\tau^{2}}-\coth\tau\dfrac{\partial f}{\partial\tau}
\end{equation}
\begin{equation}
\dfrac{\partial f}{\partial z}=\dfrac{1}{\sinh\tau}\dfrac{\partial f}{\partial\tau}
\end{equation}
\begin{equation}
z\dfrac{\partial f}{\partial z}=\coth\tau\dfrac{\partial f}{\partial\tau}
\end{equation}
\begin{equation}
(z^{2}-1)\dfrac{\partial^{2}f}{\partial z^{2}}+2z\dfrac{\partial f}{\partial z}+\left(-k^{2}+\dfrac{\left(\dfrac{1}{4}+\xi^{2}\right)}{\left(z^{2}-1\right)}\right)f=0
\end{equation}

If we take the representation of states such that:

\begin{equation}
(z^{2}-1)\dfrac{\partial^{2}f}{\partial z^{2}}+2z\dfrac{\partial f}{\partial z}+\left(\dfrac{(\nu^{2}+1)}{4}-\dfrac{k^{2}}{(z^{2}-1)}\right)f=0
\end{equation}
we will have solution
\begin{equation}
f(z)=C_{1}\mathcal{P}_{i\nu-1/2}^{k}(z)+C_{2}\mathcal{Q}_{i\nu-1/2}^{k}(z)
\end{equation}

Obviously this is permitted as it is nothing but an overall relabelling
of the same energy eigenvalue, using switched labels. We can see how
the different forms of the Mehler-Fock transform arise, as various
projections of the hyperbolic sphere into the disk or plane. We shall
now make some observations. If we take the new form of the eigenvalue
equation:

\begin{equation}
(z^{2}-1)\dfrac{\partial^{2}f}{\partial z^{2}}+2z\dfrac{\partial f}{\partial z}+\left(\dfrac{(\nu^{2}+1)}{4}-\dfrac{k^{2}}{(z^{2}-1)}\right)f=0
\end{equation}
Consulting Gradshteyn \& Ryzik, \cite{gradshteyn2014table} we have formulae 7.141.5:
\begin{lemma}
    
\begin{equation}
\int_{_{1}}^{\infty}e^{-ax}\left(\dfrac{x+1}{x-1}\right)^{\mu/2}\mathcal{P}_{\nu-1/2}^{\mu}(x)dx=\dfrac{W_{\mu,\nu}(2a)}{a}
\end{equation}
\end{lemma}
This suggests that if we insert the factor 

\begin{equation}
\int_{_{1}}^{\infty}e^{-yz}\left(\dfrac{z+1}{z-1}\right)^{k/2}\mathcal{P}_{i\nu-1/2}^{k}(x)dx=\dfrac{W_{k,i\nu}(2y)}{y}
\end{equation}
i.e., what amounts to premultiplying by a term $\left(\dfrac{z+1}{z-1}\right)^{k/2}$
and taking a Laplace transform, we will receive an interesting complementary
form of the differential equation. 
We can write:
\begin{theorem}
    
\begin{equation}
(z^{2}-1)\dfrac{\partial^{2}f}{\partial z^{2}}+2z\dfrac{\partial f}{\partial z}+\left(\dfrac{(\nu^{2}+1)}{4}-\dfrac{k^{2}}{(z^{2}-1)}\right)f=0=\mathcal{H}[f(z)]
\end{equation}
\begin{equation}
\mathcal{H}\left[\left(\dfrac{z+1}{z-1}\right)^{k/2}f(z)\right]=\left(\dfrac{z+1}{z-1}\right)^{k/2}\left[\dfrac{(\nu^{2}+1)}{4}f+2(z-k)\dfrac{\partial f}{\partial z}+(z^{2}-1)\dfrac{\partial^{2}f}{\partial z^{2}}\right]=\mathcal{H}'[f(z)]
\end{equation}
\end{theorem}
We can therefore see that, firstly, we will have a link to the confluent
hypergeometric function via this route, and secondly, that it will
be possible to transform using the rules of the Laplace transform
to derive some interesting formulae. This is a sort of automorphic
form which gives the necessary conformal mapping from one representation
to another. If we Laplace transform the differential operator, we
find:
\begin{theorem}
    
\begin{equation}
\mathcal{L}[\mathcal{H}[f]]=\int_{0}^{\infty}e^{-px}\mathcal{H}'[f(x)]dx=p^{2}\dfrac{\partial^{2}F}{\partial p^{2}}+2p\dfrac{\partial F}{\partial p}+\left(\dfrac{1}{4}(\nu^{2}+1)-2pk-p^{2}\right)F=0
\end{equation}
\end{theorem} 

with solution: 
\begin{equation}
F(p)=C_{1}\dfrac{W_{-k,i\nu/2}(2p)}{p}+C_{2}\dfrac{M_{-k,i\nu/2}(2p)}{p}
\end{equation}
Q.E.D.
    
\end{proof}

Note that we have neglected the boundary terms in this analysis as they add little to the basic argument other than some inhomogenous terms on the right hand side of the differential equation. The demonstration
of the consequences of this result shall form another paper in and
of itself. For now, we note the various representations we have developed
for the different forms of the hyperbolic plane. We can go on to develop
kernels and transition probability densities now that we have an understanding
of the basic relationships between these groups of special functions.
There is an obvious extension to be had to the Kontorovich-Lebedev
transformation, but we shall leave that for the discussion. 

\section{Recap of Results}

We began with the differential equation, which we derived from a hyperbolic
brachistochrone. It was given by the Helmholz type equation in two
dimensions, defined by:

\begin{equation}
E\Psi=-\dfrac{\partial^{2}\Psi}{\partial\tau^{2}}-\coth\tau\dfrac{\partial\Psi}{\partial\tau}+\dfrac{1}{\sinh^{2}\tau}\dfrac{\partial^{2}\Psi}{\partial\varphi^{2}}
\end{equation}
With regards to the solution $\Psi(\tau,\varphi)$, we Fourier transformed
in the second variable thus resulting in $\Upsilon(\tau,k)=\int_{-\infty}^{+\infty}e^{-ik\varphi}\Psi(\tau,\varphi)d\varphi$,
and then were able to formulate our problem purely in terms of the
solutions to the expression given by:

\begin{equation}
\dfrac{\partial^{2}\Upsilon}{\partial\tau^{2}}+\coth\tau\dfrac{\partial\Upsilon}{\partial\tau}+\left(\dfrac{k^{2}}{\sinh^{2}\tau}-E\right)\Upsilon=0
\end{equation}
We have also shown that the eigenfunctions of this expression, depending
on the geometry and sign of energy or otherwise, may be given as Whittaker
functions as in the previous section, or as conical functions of Mehler-Fock
type. The related equation for the Whittaker function is given by:
\begin{equation}
p^{2}\dfrac{\partial^{2}F}{\partial p^{2}}+2p\dfrac{\partial F}{\partial p}+\left(\dfrac{1}{4}(\nu^{2}+1)-2pk-p^{2}\right)F=0
\end{equation}
which we obtained by using a combination of Laplace transformation
and premultiplication with a conformal factor, which is given by a
Cayley transform. We have shown that these concepts descend from the
central formula for the metric in the hyperbolic plane:
\begin{equation}
ds^{2}=\dfrac{1}{4}(-d\tau^{2}+\sinh^{2}\tau d\varphi^{2})
\end{equation}
which we derived by using a clever form of the representation theory
of the quantum brachistochrone. To the author's knowledge, this is
the first true application of matching the groups generated by this
form of time optimal control theory and the motion of objects with
continuous degrees of freedom. We can expect various different relationships
to occur depending on the transformation of the input matrices, pseudo-unitary
operators and metric, which will result in a simplification of the
overall structure we associate to the representation theory. Note
the simplicity of proof compared to \cite{vilenkin1978special}; indeed, we have
begun in a different direction to their calculation and have streamlined
the approach, while achieving equivalent results. The use of differential
geometry has been invaluable in this sense, as it gives a clear and
unambiguous way in which to associate the geometry of the group with
both the parameters which define the metric and the operators which
drive the diffusion.

\section{Links with Bessel Representation}
\begin{proposition}
    A second formulation of the eigenvalue equation is given by the Kontorovich-Lebedev functions. The Mehler-Fock equation is related to this by an integral transform.
\end{proposition}
\begin{proof}
    
If we begin with the differential equation defined
by:

\begin{equation}
(z^{2}-1)\dfrac{\partial^{2}f}{\partial z^{2}}+2z\dfrac{\partial f}{\partial z}+\left(\dfrac{(\nu^{2}+1)}{4}-\dfrac{k^{2}}{(z^{2}-1)}\right)f=0=\mathcal{H}[f(z)]
\end{equation}
We have the known integral 7.142.1 \cite{gradshteyn2014table}:
\begin{equation}
\int_{_{1}}^{\infty}e^{-ax}(x^{2}-1)^{-\mu/2}\mathcal{P}_{\nu}^{\mu}(x)dx=\sqrt{\dfrac{2}{\pi}}a^{\mu-1/2}K_{\nu+1/2}(a)
\end{equation}
Once again, this is indicative of the process we shall use to derive
the relationship between the differential system described by the
the transformed solution. We expect that premultiplication followed
by Laplace transformation shall result in a new differential equation,
with the Macdonald function as kernel. As before, we write:
\begin{equation}
(z^{2}-1)\dfrac{\partial^{2}f}{\partial z^{2}}+2z\dfrac{\partial f}{\partial z}+\left(\dfrac{(\nu^{2}+1)}{4}-\dfrac{k^{2}}{(z^{2}-1)}\right)f=0=\mathcal{H}[f(z)]
\end{equation}
\begin{equation}
(z^{2}-1)^{-k/2}\mathcal{H}\left[(z^{2}-1)^{k/2}f(z)\right]
\end{equation}
\begin{equation}
    =\left[\left(\dfrac{(\nu^{2}+1)}{4}+k(k+1)\right)f+2(k+1)z\dfrac{\partial f}{\partial z}+(z^{2}-1)\dfrac{\partial^{2}f}{\partial z^{2}}\right]
\end{equation}
Laplace transforming in $z$, and ignoring the boundary conditions
as before, we find the system:
\begin{theorem}
    
\begin{equation}
p^{2}\dfrac{\partial^{2}F}{\partial p^{2}}+(2-2k)p\dfrac{\partial F}{\partial p}+(k(k-1)+\dfrac{1}{4}(\nu^{2}+1)-p^{2})F=0
\end{equation}
\begin{equation}
F(p)=p^{k-1/2}\left(C_{1}K_{i\nu/2}(p)+C_{2}I_{i\nu/2}(p)\right)
\end{equation}
\end{theorem}
Q.E.D.
\end{proof} 
We can see therefore that this will give rise to a kernel
formula, to various types of inversion formulae, to Green's functions
and other such relationships. It is simple to see how one can arrive
at the various expressions for the Kontorovich-Lebedev transformation
by using such methods. For now, we remark on the simplicity of such
a relationship between the Mehler-Fock, Whittaker and Macdonald type
functions, and how they transform under the differential operator.

\section{Composition Formula }
\begin{proposition}
    A composition formula for the Kontorovich-Lebedev functions can be evaluated using the Mehler-Fock functions and the completeness relation.
\end{proposition}
\begin{proof}

From the formula for the Mehler-Fock kernel as in \cite{sousa2017spectral}, we
have the completeness relationship:
\begin{equation}
\pi\delta(x-y)=\int_{0}^{\infty}|\Gamma\left(ip+\mu\right)|^{2}p\sinh(\pi p)\mathcal{P}_{ip-1/2}^{-\mu+1/2}\left(x\right)\mathcal{P}_{ip-1/2}^{-\mu+1/2}\left(y\right)dp
\end{equation}
We also have the integral formula which converts a conical function
into one of Macdonald type:
\begin{equation}
a^{\lambda}\sqrt{\dfrac{\pi}{2}}\int_{1}^{\infty}(x^{2}-1)^{\lambda/2-1/4}\mathcal{P}_{i\nu-1/2}^{-\lambda+1/2}(x)e^{-ax}dx=K_{i\nu}(a)
\end{equation}
Calculating a Kontorovich-Lebedev transform, we may write:

\begin{equation}
\int_{0}^{\infty}\nu\sinh(\pi\nu)\Gamma(\lambda+i\nu)\Gamma(\lambda-i\nu)K_{i\nu}(a)K_{i\nu}(b)d\nu=\int_{0}^{\infty}\nu\sinh(\pi\nu)\Gamma(\lambda+i\nu)\Gamma(\lambda-i\nu)(ab)^{\lambda}\dfrac{\pi}{2}
\end{equation}
\[
\times\int_{1}^{\infty}(x^{2}-1)^{\lambda/2-1/4}\mathcal{P}_{i\nu-1/2}^{-\lambda+1/2}(x)e^{-ax}dx\int_{1}^{\infty}(y^{2}-1)^{\lambda/2-1/4}\mathcal{P}_{i\nu-1/2}^{-\lambda+1/2}(y)e^{-by}dy
\]

\begin{equation}
=(ab)^{\lambda}\dfrac{\pi}{2}\int_{1}^{\infty}(x^{2}-1)^{\lambda/2-1/4}e^{-ax}dx\int_{1}^{\infty}(y^{2}-1)^{\lambda/2-1/4}e^{-by}dy
\end{equation}
\[
\times\int_{0}^{\infty}\nu\sinh(\pi\nu)\Gamma(\lambda+i\nu)\Gamma(\lambda-i\nu)\mathcal{P}_{i\nu-1/2}^{-\lambda+1/2}(x)\mathcal{P}_{i\nu-1/2}^{-\lambda+1/2}(y)d\nu
\]
\begin{equation}
=(ab)^{\lambda}\dfrac{\pi^{2}}{2}\int_{1}^{\infty}(x^{2}-1)^{\lambda/2-1/4}e^{-ax}dx\int_{1}^{\infty}\delta(x-y)(y^{2}-1)^{\lambda/2-1/4}e^{-by}dy
\end{equation}
\begin{equation}
=(ab)^{\lambda}\dfrac{\pi^{2}}{2}\int_{1}^{\infty}dx.(x^{2}-1)^{\lambda-1/2}e^{-(a+b)x}
\end{equation}
\begin{equation}
=(ab)^{\lambda}\dfrac{\pi^{2}}{2}\dfrac{\Gamma(\lambda+1/2)}{\sqrt{\pi}}\left(\dfrac{a+b}{2}\right)^{-\lambda}K_{\lambda}(a+b)
\end{equation}
\begin{equation}
=\dfrac{\pi^{3/2}\Gamma(\lambda+1/2)}{2}\left(\dfrac{a+b}{2ab}\right)^{-\lambda}K_{\lambda}(a+b)
\end{equation}    
\end{proof}

as required, see e.g. \cite{gradshteyn2014table}. We have invoked the Fubini-Tonelli
theorem to interchange the integrals as usual. It is simple to see
how many more integrals of a similar kind may be approached; as there
are a great deal of readily available formulae for associated Legendre
functions, this operational technique may be one way in which formulae
involving index transformations might be found.

\section{Discussion}

This paper has presented some novel methods for determining the relationships
between some well known groups of special functions. We have demonstrated
that the different representations are related to the reference frame
in which we view things; in short, the various different slices make
up the cake, the way in which we cut it defines the particular set
of eigenfunctions we must use. We postpone discussion of the eigenfunction
representations for a following paper, as this note merely illustrates
the existence of the link and not its use in calculation. We have
shown that the theories of brachistochrones are well suited to the
application of ideas from differential geometry and the calculus of
diffusions; this is expected to be a straightforward way in which
particle dynamics can be understood, even for complex systems as we
have considered in the hyperbolic plane.

We have shown directly how one can move from the group representation,
given as a matrix with some symmetry, to the continuous operators
familiar from differential calculus, via the methods of differential
geometry. This is to be an area of future research, as the areas and
topics developed in understanding this link will have direct bearing
on questions of a physical nature. For now, we merely comment that
the hyperbolic plane is not just a hypothetical construct. The differential
systems we have developed an understanding of by using these techniques
have direct bearing on such varied problems as oscillatory quantum
systems, pricing of exotic options and various questions in hyperbolic
flow. It will be interesting to see whether there are higher dimensional
analogues of these types of problems.

For simplicity, this note has not addressed the form of the constants
associated with the solutions of the various differential equations.
One notes that one possible way in which to consistently derive these
parameters is by using the Titsmarsch-Kodaira theorem as in \cite{sousa2017spectral}.
The various expressions for the normalisations and some brief calculations
are given in the appendix.

\section{Future Directions}

It is interesting to contemplate which other groups would be amenable
to the treatment we have prescribed to SU(1,1). The most obvious of
these seems to be SU(2), with consequences for quantum mechanical
systems. Further research directives would appear to be establishing
the relevant sets of special functions for efficient computation of
exotic options pricing using such methods, as well as determining
whether such techniques can be used for pricing basic contracts using
such things as CIR and OU processes. It would be hoped that a general
theory of (pseudo)-spherical systems might be developed using such
intuitions.

In other possible extensions, it might be possible to use these techniques
to establish the nature of diffusion on higher dimensional hyperbolic
groups of a physical nature. As both the equations of quantitative
finance and physics are based around similar methods, theory and practice,
it is possible to see that there may be new methods arising in this
space. Other outstanding questions of interest might include the computation
of Wigner functions using such methods for SU(1,1), see e.g. \cite{seyfarth2020wigner}
for an approach based on creation and annihilation operators. Such
functions have a natural interpretation in statistical mechanics,
giving the Poisson bracket of the Hamiltonian flow.

This paper has not discussed the problem of the discrete states, or
the difference between the scattered and bound states. We note that
for certain parts of the spectrum that this paper has not considered,
different groups of special functions may be more important than the
continuous groups we have calculated in this note, in particular the
Laguerre functions. These expansions have been used to evaluate exponential
Brownian motion \cite{pintoux2010direct, pintoux2011dothan} and the pricing of Asian
options \cite{zhang2010path}. In particular, of interest is the development
of kernel densities for these type of systems, and finding efficient
ways to generate numerical values for the confluent hypergeometric
function.

\section*{Acknowledgment}
This project was supported under ARC Research Excellence Scholarships at the University of Technology, Sydney. The author acknowledges useful discussions and support from Dr. Mark Craddock and Prof. Anthony Dooley.

\section*{Conflicts of Interest}
The author has no conflicts of interest to declare that are relevant to the content of this article.

\bibliography{sn-bibliography}

\end{document}